\documentclass{article}

\usepackage[utf8]{inputenc}
\usepackage{amsmath,amsthm,amssymb}
\usepackage[margin=3cm]{geometry}
\usepackage{import}
\usepackage{graphicx}
\usepackage{tikz}
\usepackage{tikz-cd}
\usepackage{hyperref}
\usepackage{cleveref}
\usepackage{enumerate}
\usepackage{adjustbox}
\usepackage{array}
\usepackage{tabularx}
\usepackage{multirow}
\usepackage{hhline}
\usepackage{authblk}
\usepackage{titlesec}
\usepackage{caption} 
\usepackage{tabu}
\usepackage{algorithm}
\usepackage[noend]{algpseudocode}
\def\algbackskip{\hskip-\ALG@thistlm}

\newtheorem{theorem}{Theorem}[section]
\newtheorem{corollary}{Corollary}[theorem]
\newtheorem{lemma}[theorem]{Lemma}
\newtheorem{proposition}[theorem]{Proposition}
\theoremstyle{definition}
\newtheorem{definition}[theorem]{Definition}
\newtheorem{remark}[theorem]{Remark}

\titlespacing*{\section}
{0pt}{5.5ex plus 1ex minus .2ex}{4.3ex plus .2ex}
\titlespacing*{\subsection}
{0pt}{5.5ex plus 1ex minus .2ex}{4.3ex plus .2ex}

\newcommand{\f}[1]{f_{\restriction_{#1}}}
\newcommand{\fO}{\f{\Omega}}
\newcommand{\fT}{\f{\Theta}}
\newcommand{\B}{\mathbb{B}}
\newcommand{\A}{\mathcal{A}}

\title{Control in Boolean networks with model checking}

\author[1,2]{Laura Cifuentes-Fontanals} 
\author[1]{Elisa Tonello} 
\author[1]{Heike Siebert}
\affil[1]{Freie Universität Berlin, Germany}
\affil[2]{Max Planck Institute for Molecular Genetics, Berlin, Germany}

\date{}

\begin{document}

\maketitle

\subsubsection*{Abstract}

Understanding control mechanisms in biological systems plays a crucial role in important applications, for instance in cell reprogramming. Boolean modeling allows the identification of possible efficient strategies, helping to reduce the usually high and time-consuming experimental efforts. Available approaches to control strategy identification usually focus either on attractor or phenotype control, and are unable to deal with more complex control problems, for instance phenotype avoidance. They also fail to capture, in many situations, all possible minimal strategies, finding instead only sub-optimal solutions. In order to fill these gaps, we present a novel approach to control strategy identification in Boolean networks based on model checking. The method is guaranteed to identify all minimal control strategies, and provides maximal flexibility in the definition of the control target. We investigate the applicability of the approach by considering a range of control problems for different biological systems, comparing the results, where possible, to those obtained by alternative control methods.


\section{Introduction}

The study of control of cellular systems has opened multiple possibilities for application in bioengineering and medicine. It also provides the possibility to make predictions on model behaviour, for instance about the reachability of phenotypes under certain mutations, that could be verified experimentally and used for model validation. Experimental approaches for the identification of effective control targets are usually costly and time consuming. To help reducing these efforts, mathematical modelling can be used to identify, in silico, potentially useful interventions that could lead to the reduction of experimental trials \cite{Sinergies}. 

Boolean modeling is often used to model biological systems, since it is able to capture qualitative behaviours by describing the activating or inhibiting interactions between different species using logical functions. The species are represented by binary-valued nodes, whose two activity levels might indicate for instance in a gene-regulatory network whether a certain gene is expressed or not. The simplicity of the Boolean formalism helps coping with the usual problem of lack of parameter information when modeling biological processes while capturing the relevant dynamics of biological systems \cite{CellFate_network,MAPK_network,TLGL_network}.

In the context of control for drug target identification or cell reprogramming, the main goal is the identification of controls that require a minimal number of system interventions. Providing multiple alternatives for minimal control interventions is also desirable, so that suitable interventions for experimental implementation can be found. Furthermore, there are many different scenarios and goals to which control might be applied, for instance, to enforce or avoid a specific behaviour in a biological system. An example of such a scenario could be a cell differentiation system where a particular cell type is to be avoided since it can be linked to the development of cancer or another pathology \cite{avoid_hybrid}.

Many approaches have been developed for control of biological systems, covering different contexts and goals. Some of them focus on leading the system to an attractor of interest, starting from a specific initial state \cite{ControlBasins} or from any possible initial state \cite{ControlMotifs}. This control problem is known as attractor control. However, in some cases, a small number of observable and measurable components is sufficient to capture the relevant features of the system attractors, for example the set of biomarkers defining a phenotype. In such cases, it might be useful to aim the control towards the phenotype defined by these biomarkers rather than a specific attractor, since fewer interventions might be sufficient. This approach, which targets a set of relevant variables instead of a specific attractor, is known as target control. Several methods have been developed for such control problems \cite{InterventionSets, ControlBCN} using different computation techniques.

A basic approach to control is value percolation, which is also a core step in many more sophisticated methods \cite{InterventionSets,TargetControl}. Approaches based on value percolation can be implemented efficiently \cite{control_asp}. However, they are quite restrictive and might miss many possible control strategies. A step towards the identification of some of these missed control strategies using trap spaces was presented in \cite{CS_via_trapspaces}. Although this approach is more flexible than just value percolation, it also does not identify all the possible control strategies. In the last years, multiple methods have been developed for control strategy identification, looking for instance at the stable motifs of the system \cite{ControlMotifs} or exploiting computational algebra methods \cite{ControlAlgebra}. These approaches are usually focused on targeting an attractor or subspace and they also do not generally uncover all possible minimal control strategies. In order to bridge this gap, recent works have tackled the problem of attractor control by using basins of attraction, sets of states from which only a specific attractor is reachable \cite{cabean}. Such approaches increase in many cases the amount of strategies identified. However, they are still limited to control for attractors and lack flexibility to deal with groups of attractors or phenotypes as well as with attractor avoidance. To the best of our knowledge, there is no method that can identify all the optimal control strategies for a general set of states or attractors.

In this work, we introduce a new approach for control strategy identification that provides a complete solution set of minimal controls and allows full flexibility in the control target. Identifying all the minimal control strategies for a general set of states is a complex problem. It might require the full exploration of the state space, which grows exponentially with the size of the network. To deal with this computational explosion, we explore model checking techniques. Model checking is a verification method that allows to determine whether a transition system satisfies a specific property. Although originating in the field of computer science, model checking has been successfully applied to analyse biological networks and a wide variety of tools have been developed \cite{modcheck_overview}. Model checking presents many advantages, for instance the use of symbolic representation, which allows to deal with systems with a large number of states and other problems that could not be handled otherwise. Yet, tackling a wider and more complex control problem naturally entails higher computational costs, since many shortcuts and reduction methods do not apply. Therefore, we investigate efficient preprocessing techniques that can be used to significantly reduce the computational cost and make it suitable for application. 

As mentioned above, this work presents a model checking-based method to identify optimal control strategies for any target subset. We start with a general overview about Boolean modeling and model checking (\Cref{Background}), focusing on the main concepts used in this work. Then we introduce the formal definition of control strategy, present some properties of value percolation that are used in our approach and establish the basis for the control strategy computation with model checking (\Cref{CS}). The implementation of our approach is detailed in \Cref{Computation}, with ideas to reduce the search space size and improve the performance of the method. Finally, in \Cref{Application} we show the applicability of our method to different biological networks and compare our results with existing control approaches.


\section{Background} \label{Background}

\subsection{Boolean networks and dynamics}

We define a \emph{Boolean network} as a function $f \colon \B^n \rightarrow \B^n$, with $\B = \{0,1\}$. The set of variables or components of $f$ is denoted by $V = \{1,\dots,n\}$. Given a Boolean function different dynamics can be defined depending on the way components are updated. A dynamics is usually represented by the \emph{state transition graph} (STG), a graph whose set of vertices is the state space $\B^n$ and whose edges represent the transitions between them. The \emph{synchronous dynamics} $SD(f)$ defines transitions that update at the same time all the components that can be updated. Thus, the synchronous state transition graph has an edge from $x \in \B^n$ to $y \in \B^n$ if and only if $x \neq y$ and $y = f(x)$. In order to better capture the different times scales that might coexist in a biological system, the \emph{asynchronous dynamics} $AD(f)$ is often used. It defines transitions that update only one component at a time. Therefore, its state transition graph has an edge from $x \in \B^n$ to $y \in \B^n$ if there exists $i \in V$ such that $y_i = f_i(x) \neq x_i$ and $y_j = x_j$ for all $j \neq i$. The \emph{general asynchronous dynamics} $GD(f)$ generalises the two previous ones by defining transitions that update a non-empty subset of components. Thus, given $x,y \in \B^n$ there is a transition from $x$ to $y$ if there exists a subset $\emptyset \neq I \subseteq V$ such that $y_i = f_i(x) \neq x_i$ for all $i \in I$ and $y_j = x_j$ for all $j \notin I$. To simplify the notation, we use $D(f)$ to refer to any of these three dynamics. 

A \emph{path} in an STG is defined as a sequence of nodes $x^0, x^1, \ldots$ such that there exists an edge from $x^{i-1}$ to $x^i$ for all $i \geq 1$. We denote the set of paths starting in a state $x$ as $\mathrm{Paths}(x)$. Given a state $x \in \B^n$, we define $\mathrm{Reach}(x) = \{ y \in \B^n$ $|$ $\exists \pi \in \mathrm{Paths}(x)$ s.t. $y \in \pi\}$ and given $S \subseteq \B^n$, $\mathrm{Reach}(S)$ is the set $\{y \in \B^n \ | \ y \in \mathrm{Reach}(x)$ for some $x \in S\}$. Note that $x \in \mathrm{Reach}(x)$, since $x$ is the 1-element path to $x$. A set $T \subseteq \B^n$ such that $T = \mathrm{Reach}(T)$ is called a \emph{trap set}. An \emph{attractor} is a minimal trap set under inclusion. Attractors correspond to the minimal strongly connected components of the STG and they might vary in different updates. In biological systems, attractors consisting of only one state (\emph{steady states}) might correspond to different cell fates, while attractors consisting of more than one state (\emph{cyclic attractors}) might be associated with sustained oscillation. \Cref{fig:ex_mcheck} shows an example of asynchronous dynamics with a steady state and a cyclic attractor.

Given $I \subseteq V$ and $c \in \B^n$ we define the \emph{subspace} induced by $I$ and $c$ as the set $\Sigma(I,c) = \{ x \in \B^n$ $|$ $x_i = c_i$ $\forall i \in I\}$. We denote a subspace by writing the value 0 or 1 for variables that are fixed and ${*}$ for the free variables. For example, $0{*}{*}1$ denotes the subspace fixing the first variable to 0 and the fourth to 1, that is, $S = \{x \in \B^n$ $|$ $x_1 = 0, x_4 = 1 \}$. Thus, a subspace can be seen as an element of $\{0,1,*\}^n$. Given $S \in \{0,1,*\}^n$, $S_i$ denotes the value of the component $i$. A \emph{trap space} is a subspace that is also a trap set. Trap spaces, contrary to attractors or trap sets, are the same in any update.

In this work we consider interventions that fix certain components to specific values. Note that a set of such interventions can be seen as a subspace and the effect that these interventions have on the dynamics can be described by restricting the original Boolean function to the intervention subspace. Given a Boolean function $f$ and a subspace $\Theta = \Sigma(I,c)$ we define the restriction of $f$ to the subspace $\Theta$ as:

$$
\fT\colon \Theta \rightarrow \Theta,
\text{ where for all } i \in V \text{, } (\fT)_i(x) = \left\{
\begin{array}{ll}
f_i(x), & i \notin I, \\
c_i, &  i \in I. \\
\end{array}
\right.
$$

Note that all the definitions above can be applied to the restriction by identifying $\fT$ with a Boolean function on $\B^{n - |I|}$. An example of the dynamics of a Boolean function restricted to a subspace is shown in \Cref{ex:cs_comp}.

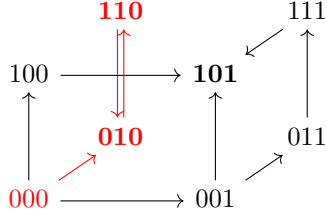
\begin{figure}
\begin{center}
\begin{tikzcd}[column sep=10, row sep=10]
 & {\color{red} \textbf{110}} \arrow[dd, shift right=0.5, red] & & 111 \arrow[dl] \\
100 \arrow[rr] & & \textbf{101} & \\
 & {\color{red} \textbf{010}} \arrow[uu, shift right=0.5, red] & & 011 \arrow[uu] \\
{\color{red}000} \arrow[rr] \arrow[uu] \arrow[ur, red] & & 001 \arrow[ur] \arrow[uu]
\end{tikzcd}
\end{center}
\caption{Asynchronous state transition graph of a Boolean function with two attractors, $\{101\}$ and $\{010,110\}$, and six trap spaces: $101$, ${*}10$, $10{*}$, $1{*}1$, ${*}{*}1$ and ${*}{*}{*}{*}$. The states 001, 011, 101 and 111 satisfy the state formula \textbf{AG}$\phi_3$, where $\phi_i(x) := (x_i = 1)$, while \textbf{EF}$\phi_3$ is satisfied by all the states except $010$ and $110$. The path $\pi$ that starts at 000 and then oscillates between 010 and 110 (in red) satisfies for instance \textbf{F}$\phi_2$ and \textbf{G}$\neg \phi_3$ but not \textbf{G}$\phi_1$.}
\label{fig:ex_mcheck}
\end{figure}


\subsection{Model checking}

This section provides a practical introduction to the model checking concepts required for the description of our approach. For a more extensive and detailed explanation of model checking we refer the reader to~\cite{book_modcheck}. Model checking is a formal method used in computer science to solve verification problems. Its application to the control strategy problem presents many advantages, for instance the use of symbolic representation, which allows to deal with systems with a large number of states, like STGs of Boolean networks. Moreover, many efficient algorithms have been developed and are available for running model checking queries. An overview of existing model checking tools in the context of biochemical networks analysis can be found in \cite{modcheck_overview}. 

Model checking allows to verify whether a given transition system satisfies a specific property. A transition system is defined as a set of states and a set of transitions, which represent changes from one state to another. Formally, a labeled transition system (LTS) is defined by a tuple $(\mathsf{S,T,L})$ where $\mathsf{S}$ is a finite set of states, $\mathsf{T} \subseteq \mathsf{S} \times \mathsf{S}$ is a transition relation such that $(x^1, x^2) \in \mathsf{T}$ if there exists a possible transition from state $x^1$ to state $x^2$ and $\mathsf{L} \colon \mathsf{S} \rightarrow 2^{AP}$ is a labeling function with $AP$ a finite set of atomic propositions. In the following, a transition $(x^1, x^2)$ will also be denoted by $x^1 \rightarrow x^2$. The labeling function $\mathsf{L}$ gives a set $\mathsf{L}(x) \in 2^{AP}$ of atomic propositions for each state $x$ which includes exactly the atomic propositions satisfied by $x$. In the Boolean context, an STG defines an LTS, where the set of states is $\B^n$ and the transitions are defined by the Boolean function and the type of update that is chosen. For our purposes, we need a deadlock-free transition system, so we add extra transitions $(x,x) \in \mathsf{T}$ for every steady state $x \in \B^n$. We use the atomic propositions $AP = \{ (v = c)$ $|$ $v \in V, c \in \B \}$ and define the labeling function by $(v = c) \in \mathsf{L}(x)$ if and only if $x_v = c$.

There are different ways to express properties of a transition system. In our case, we use Computational Tree Logic (CTL). CTL is based on a branching notion of time, where the behavior of the system is represented by a tree of states. In the case of Boolean networks, one can imagine that every path starting in a state $x \in \B^n$ is represented as a branch in a tree rooted in $x$. In the following we introduce the main concepts of CTL.

We distinguish between state properties and path properties. In this context, a \emph{path} is an infinite sequence  $x^0, x^1, \ldots \in \mathsf{S}$ such that $(x^{i-1},x^i) \in \mathsf{T}$ for all $i \geq 1$. A statement about a state or a path can be made using a CTL formula. A CTL \emph{state formula} $\phi$ over the set of atomic propositions $AP$ is of the form
$$
\phi := a \text{ $|$ } \neg \phi \text{ $|$ } \phi_1 \lor \phi_2 \text{ $|$ } \phi_1 \land \phi_2 \text{ $|$ } \textbf{E} \varphi \text{ $|$ } \textbf{A} \varphi
$$

where $a \in AP$ is an atomic proposition, \textbf{E} is the \emph{exists} operator, \textbf{A} is the \emph{for all} operator, $\phi$, $\phi_1$ and $\phi_2$ are CTL state formulas and $\varphi$ is a CTL \emph{path formula}, which in our work will be of the form:

$$
\varphi := \textbf{F} \psi \text{ $|$ } \textbf{G} \psi
$$
 
where \textbf{F} is the \emph{future} operator, \textbf{G} the \emph{global} operator and $\psi$ a CTL state formula. See \Cref{tab:operators} for
the satisfaction relation $\models$ for transition systems and CTL formulas. We say that a CTL state formula $\phi$ \emph{is satisfied by} a state $x$ if and only if $x \models \phi$, that is, $\phi(x) = true$. Analogously, a CTL path formula $\varphi$ \emph{is satisfied by} a path $\pi$ if and only if $\pi \models \varphi$. \Cref{fig:ex_mcheck} shows some examples of state and path formulas which are satisfied in an STG.

\begin{table}
$$
\begin{array}{lll}
x \models true & \\
x \models a & \text{ iff } & a \in \mathsf{L}(x) \\
x \models \phi_1 \lor \phi_2 & \text{ iff } & x \models \phi_1 \lor x \models \phi_2 \\
x \models \phi_1 \and \phi_2 & \text{ iff } & x \models \phi_1 \land x \models \phi_2 \\
x \models \neg \phi & \text{ iff } & x \nvDash \phi \\
x \models \textbf{E}\varphi & \text{ iff } & \exists \pi \in \mathrm{Paths}(x) \text{ s.t. } \pi \models \varphi \\
x \models \textbf{A}\varphi & \text{ iff } & \forall \pi \in \mathrm{Paths}(x), \pi \models \varphi \\
\pi \models \textbf{F}\phi & \text{ iff } & \exists y \in \pi \text{ s.t. } y \models \phi \\
\pi \models \textbf{G}\phi & \text{ iff } & \forall y \in \pi, y \models \phi \\
x \models \textbf{EF}\phi & \text{ iff } & \exists \pi \in \mathrm{Paths}(x), \exists y \in \pi  \text{ s.t. } y \models \phi \\
x \models \textbf{AG}\phi & \text{ iff } & \forall \pi \in \mathrm{Paths}(x), \forall y \in \pi, y \models \phi \\
\end{array}
$$
\caption{Satisfaction relation and semantics for CTL formulas, with $a \in AP$ an atomic proposition, $x \in S$ a state, $\pi$ a path in the transition system, $\varphi$ a path formula and $\phi$, $\phi_1$ and $\phi_2$ state formulas.}
\label{tab:operators}
\end{table}


\section{Control strategies} \label{CS}

A control strategy is defined as a set of interventions that lead the controlled system to a target subset. This target subset usually represents a specific stable behaviour, for instance an attractor or a subspace representing a phenotype. The formal definition of a control strategy is given below.

\begin{definition}
\label{def:cs}
Given a Boolean network $f$ and a subset $P \subseteq \B^n$, a \emph{control strategy for the target $P$} in $D(f)$ is a subspace $\Theta \subseteq \B^n$ such that, for any attractor $\A$ of $D(\fT)$, $\A \subseteq P$.
\end{definition}

In other words, a subspace $\Theta = \Sigma(I,c)$ is a control strategy for a subset $P$ if fixing the variables in $I$ to their corresponding values in $c$ forces the system to evolve to the desired target $P$.
The size of a control strategy $\Theta = \Sigma(I,c)$ is defined as the size of $I$, and is therefore the number of interventions. Optimal control strategies are the ones with the lowest number of interventions, that is, the maximal subspaces with respect to inclusion. An example of a control strategy is shown in \Cref{ex:cs_comp}, where fixing the variable $x_3 = 0$ is enough to guarantee that the system will evolve to the target steady state $110$.

Note that \Cref{def:cs} considers a subset as the control target, encompassing both attractor control and target control. Moreover, it provides the flexibility to deal with further control problems, such as control to union of attractors ($P = \bigcup_i \A_i$) or attractor avoidance ($P = \B^n \backslash \A$).

\begin{figure}
\begin{minipage}{0.15\linewidth}
\flushright \textbf{(a)}
\end{minipage}
\begin{minipage}{0.35\linewidth}
\begin{center}
\tikz[overlay]{
\filldraw[fill = red!20,draw=black, thick, dotted] (0.1,1.35) rectangle (2.03,-1.2);
\filldraw[fill = red!10, draw=red, thick] (1.25,0.88) rectangle (1.93,1.25);
\filldraw[fill = white, draw=gray] (3.3,-0.42) rectangle (4,-0.05);}
\begin{tikzcd}[column sep=5, row sep=5]
 & {\color{red} \textbf{110}} & & {\color{gray}111} \arrow[dd, gray] \\
{\color{red}100} \arrow[ur, red] & & {\color{gray} 101} \arrow[ll, gray] & \\
 & {\color{red}010} \arrow[ld, red] & & {\color{gray}\textbf{011}} \\
{\color{red}000} \arrow[uu, red] \arrow[rr, gray] & & {\color{gray} 001} \arrow[uu, gray] \\
\end{tikzcd}
\end{center}
\end{minipage}
\begin{minipage}{0.15\linewidth}
\flushright \textbf{(b)}
\end{minipage}
\begin{minipage}{0.2\linewidth}
\begin{center}
\tikz[overlay]{
\filldraw[fill = red!10, draw=red, thick] (1.4,0.7) rectangle (2.05,0.3);}
\begin{tikzcd}[column sep=10, row sep=5]
100 \arrow[r] & {\color{red} 110} \\
& \\
000 \arrow[uu] & 010 \arrow[l] \\
\end{tikzcd}
\end{center}
\end{minipage}
\caption{(a) Asynchronous dynamics of the Boolean function $f(x) = (\bar{x}_2 \lor x_1 \bar{x}_3, x_1 \bar{x}_3 \lor x_2 x_3, \bar{x}_1 \bar{x}_2 \lor x_2 x_3)$ with two steady states $110$ and $011$. (b) Asynchronous dynamics of the restriction of $f$ to $\Omega = {*}{*}0$, $\fO(x) = (\bar{x}_2 \lor x_1, x_1, 0)$. $\Omega$ is a control strategy for $P = 110$ in $AD(f)$. $\Omega$ does not percolate to $P$ nor to any non-trivial trap space. $T = 110 \subseteq P$ is the only minimal trap space of $\fO$ and is complete in $D(\fO)$.}
\label{ex:cs_comp}
\end{figure}
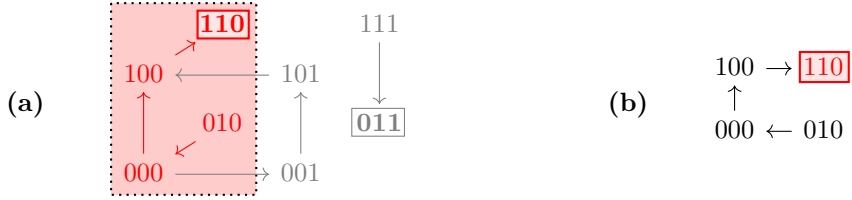


\subsection{Percolation}\label{subsec:perc}

In this subsection we introduce formally the concept of percolation, used in many approaches to control. We also deduce properties of percolated subspaces that are useful for control strategy identification and that are used later in our approach.

\begin{definition}
Given a Boolean function $f$, the \emph{percolation function} with respect to $f$ is the function
$F(f) \colon \{0,1,{*}\}^n \rightarrow \{0,1,{*}\}^n$ defined as follows. For each subspace $S \subseteq \B^n$,
$$
F(f)_i(S) = \left\{
\begin{array}{ll}
0 & \text{if } f_i(x) = 0 \text{ for all } x \in S, \\
1 & \text{if } f_i(x) = 1 \text{ for all } x \in S, \\
{*} &  otherwise. \\
\end{array}
\right.
$$
\label{def:perc}
\end{definition}

\begin{definition}
Let $f$ be a Boolean function and $S,S' \subseteq \B^n$ two subspaces. We say that the subspace \emph{$S$ percolates to $S'$ under $f$} if and only if there exists $k \geq 0$ such that $F(f)^k(S) = S'$.
\end{definition}

For any trap space $T$ and its image $T' = F(f)(T)$, we have $T' \subseteq T$, since by definition $F$ preserves the fixed components of $T$. The free components in $T$ might get fixed or remain free depending on the Boolean function $f$. In fact, $T'$ is a trap space of $f$, since for any fixed variable $i \in V$ in $T'$, $f_i(x) = T'_i$ by definition of $F$. Moreover, $F(f)^k(T)$ is a trap space for any $k \geq 0$.

\begin{remark}
\label{res_perc_ts}
Let $f$ be a Boolean function and $S \subseteq \B^n$ a subspace. Let $k \geq 0$ and $S^k = F(\f{S})^k(S)$. Then $S^k$ is a trap space of $\f{S}$ for every $k \geq 0$.
\end{remark}

Note that the paths in the dynamics of $F(f)$ starting at a trap space $T$ cannot have cycles and, consequently, all the reachable attractors from $T$ in these dynamics are fixed points. When considering the synchronous dynamics of $F(f)$, each initial trap space $T$ leads to a unique fixed point.

\begin{definition}\label{def:perc_subs}
Given a Boolean function $f$ and a trap space $T$, we call the unique fixed point $T'$ of the synchronous dynamics of $F(f)$ reachable from $T$ the \emph{percolated subspace from $T$ with respect to $f$}, that is, $T' = F(f)^{k}(T)$ with $k$ such that $F(f)^{k}(T) = F(f)^r(T)$ for all $r \geq k$.
\end{definition}

In order to relate value percolation to control strategies, we derive some dynamical properties of percolated subspaces.

\begin{lemma}
\label{perc_path}
Let $f$ be a Boolean function, $T \subseteq \B^n$ a trap space. Let $k \geq 0$ and $T^k = F(f)^k(T)$. Then for every $x \in T$ there exists a path in $D(f)$ from $x$ to some $y \in T^k$.  
\end{lemma}

\begin{proof}
It is enough to show that if $T$ is a trap space, then for every $x \in T$ there exists a path in $D(f)$ from $x$ to some $y \in F(f)(T)$. Set $T' = F(f)(T)$, with $I' \subseteq V$ being the set of fixed variables in $T'$. Since $T' \subseteq T$, for all $i \in I'$, $f_i(x) = T'_i$ by definition of $F$. Now let us look at every update separately. If $D = AD$, for every $i \in I'$, $x$ admits a successor $y$ in $AD(f)$ with $y_i = T'_i$ and $y_j = x_j$ for $j \neq i$; therefore there exists a path from any state in $T$ to $T'$. If $D = SD$, $f_i(x) = T'_i$ for all $i \in I'$ and so $x$ admits a successor $y \in T'$. The case $D = GD$ follows from the other cases, since all the paths in $SD(f)$ or $AD(f)$ are also paths in $GD(f)$.
\end{proof}

\begin{corollary}
\label{res_perc_path}
Let $f$ be a Boolean function, $S \subseteq \B^n$ a subspace. Let $k \geq 0$ and $S^k = F(\f{S})^k(S)$. Then for every $x \in S$ there exists a path in $D(\f{S})$ from $x$ to some $y \in S^k$.  
\end{corollary}

\begin{lemma}
\label{perc_attr}
Let $f$ be a Boolean function and $S \subseteq \B^n$ a subspace. Let $k \geq 0$ and $S^k = F(\f{S})^k(S)$. Then $\A$ is an attractor of $\f{S}$ if and only if $\A \subseteq S^k$ and $\A$ is an attractor of $\f{S^k}$.
\end{lemma}

\begin{proof}
As noted in \Cref{res_perc_ts}, $S^k$ is a trap space of $\f{S}$. Then $\f{S}(x) = \f{S^k}(x)$ for all $x \in S^k$. Therefore, any attractor of $\f{S^k}$ is also an attractor of $\f{S}$ and if $\A$ is an attractor of $\f{S}$ and $\A \subseteq S^k$, then $\A$ is also an attractor of $\f{S^k}$.
Let $\A$ be an attractor of $\f{S}$. Then $\A = \mathrm{Reach}_{\f{S}}(\A)$. By \Cref{res_perc_path}, for every $x \in S$ there exists a path in $\f{S}$ from $x$ to some $y \in S^k$. Therefore, $\mathrm{Reach}_{\f{S}}(\A) \cap S^k \neq \emptyset$ and so, $S^k \cap \A \neq \emptyset$. Since $S^k$ is a trap space of $\f{S}$, $\A \subseteq S^k$. Therefore all the attractors of $\f{S}$ are contained in $S^k$ and, consequently, are also attractors of $\f{S^k}$.
\end{proof}

\begin{corollary}
\label{perc_cs}
Let $f$ be a Boolean function and $S \subseteq \B^n$ a subspace and $P \subseteq \B^n$ a subset. Let $k \geq 0$ and $S^k = F(\f{S})^k(S)$. $S$ is a control strategy for $P$ if and only if $S^k$ is a control strategy for $P$.
\end{corollary}

\Cref{perc_cs} provides a way to identify control strategies or discard candidate subspaces by using value percolation. Moreover, checking whether the percolated subspace satisfies the conditions of \Cref{def:cs} instead of the original subspace allows to reduce the dimension of the restricted network and, consequently, to simplify the verification problem.

Percolation-only methods select candidate subspaces and percolate them. If the resulting subspace is contained in the target subspace, the candidate subspace is identified as a control strategy \cite{InterventionSets}. This type of control strategies that can be identified efficiently \cite{control_asp}. However, additional control strategies might exist, that do not directly percolate to the target subspace, as shown in \cite{CS_via_trapspaces}, or even to non-trivial intermediate trap spaces (see \Cref{ex:cs_comp}). With our approach, value percolation can also be exploited as a first step towards control strategy identification, as a means to achieve dimensionality reduction, as described by \Cref{perc_cs}.


\subsection{Completeness}\label{sub:completeness}

To improve control detection, we propose to define sufficient conditions on the system restricted to a candidate subspace to identify this candidate as a control strategy. Moreover, conditions on minimal trap spaces can be defined in order to deduce properties of the system attractors, in particular, their belonging to a target subset.

\begin{definition}
A set of trap spaces $\mathcal{T}$ is \emph{complete} in $D(f)$ if and only if for every attractor $\A$ of $D(f)$ there exists $T \in \mathcal{T}$ such that $\A \subseteq T$. A Boolean function $f$ is \emph{complete} if its minimal trap spaces are complete.
\end{definition}

Completeness of the minimal trap spaces has been used for attractor approximation and it can be detected using model checking as described in \cite{AttractorApprox_Klarner}. The following proposition presents sufficient conditions for a subspace to be a control strategy. Given a candidate subspace, if the set of minimal trap spaces of the restricted system is complete and contained in the target subset, then the candidate subspace is a control strategy for that subset.

\begin{proposition}\label{cs:completeness}
Let $f$ be a Boolean function, $P,\Theta \subseteq \B^n$ subspaces and $\mathcal{T}$ a set of trap spaces of $\fT$. If all the trap spaces of $\mathcal{T}$ are contained in $P$ and $\mathcal{T}$ is complete for $D(\fT)$, then $\Theta$ is a control strategy of $P$.
\end{proposition}

\begin{proof}
Let $\mathcal{A}$ be an attractor of $D(\fT)$. Since $\mathcal{T}$ is complete for $D(\fT)$, there exists a minimal trap space $T \in \mathcal{T}$ such that $\mathcal{A} \subseteq T$. Therefore, $\mathcal{A} \subseteq T \subseteq P$.
\end{proof}

\Cref{cs:completeness} provides sufficient conditions that allow to identify new control strategies missed by percolation-based approaches (see example in \Cref{ex:cs_comp}). We refer to this approach for control strategy identification as the \emph{completeness} approach. However, it still does not characterize all the possible control strategies satisfying \Cref{def:cs}, as can be seen in \Cref{ex:notcomplete}. To obtain the full solution set, we formulate a model checking approach, as shown in the next section.

\begin{figure} 
\begin{minipage}{0.05\linewidth}
\centering \textbf{(a)}
\end{minipage}
\begin{minipage}{0.6\linewidth}
\begin{center}
\begin{tikzcd}[column sep=5]
0110 \arrow[ddd, shift right=1] & & & 0111 \arrow[ddd] \arrow[dl, shift right=1] \arrow[rrrrr, bend left=25] & & 1110 \arrow[ddd] \arrow[dr] \arrow[rrr] & & & \textbf{\color{blue}1111} \\
& {\color{red}0010} \arrow[d, shift right=1, red] & 0011 \arrow[rrrr, bend right=25] \arrow[ur, shift right=1] & & & & 1010 \arrow[d, shift right=1] & 1011 \arrow[d] \arrow[ru]  & \\
& {\color{red}0000} \arrow[u, shift right=1, red] \arrow[r, shift right=1, red] & {\color{red}0001} \arrow[l, shift right=1, red] & & & & 1000 \arrow[r] \arrow[u, shift right=1] & 1001 \arrow[rd]  & \\
0100 \arrow[uuu, shift right=1] \arrow[ur] & & & 0101 \arrow[lll] & & 1100 \arrow[rrr] \arrow[ru] & & & 1101 \arrow[uuu] \arrow[lllll, bend left=25] \\
\end{tikzcd}
\end{center}
\end{minipage}
\begin{minipage}{0.05\linewidth}
\centering \textbf{(b)}
\end{minipage}
\begin{minipage}{0.25\linewidth}
\vspace{0.5cm}
\begin{center}
\begin{tikzcd}[column sep=5]
0110 \arrow[ddd, shift right=1] & & & 0111 \arrow[ddd] \arrow[dl, shift right=1]  \\
& {\color{red}0010} \arrow[d, shift right=1, red] & 0011 \arrow[ur, shift right=1] & \\
& {\color{red}0000} \arrow[u, shift right=1, red] \arrow[r, shift right=1, red] & {\color{red}0001} \arrow[l, shift right=1, red] & \\
0100 \arrow[uuu, shift right=1] \arrow[ur] & & & 0101 \arrow[lll] \\
\end{tikzcd}
\end{center}
\end{minipage}
\caption{The asynchronous dynamics of the Boolean function $f$ and $\f{\Omega}$, with $\Omega = 0{*}{*}{*}$, are represented in (a) and (b) respectively. $\Omega$ is a control strategy for $P = 00{*}{*}$ in $AD(f)$. $\Omega$ is not a trap space nor percolates to any trap space. There is no non-trivial trap space in $\fO$ and therefore, $\Omega$ would not be identified as control strategy by the completeness approach.}
\label{ex:notcomplete}
\end{figure}
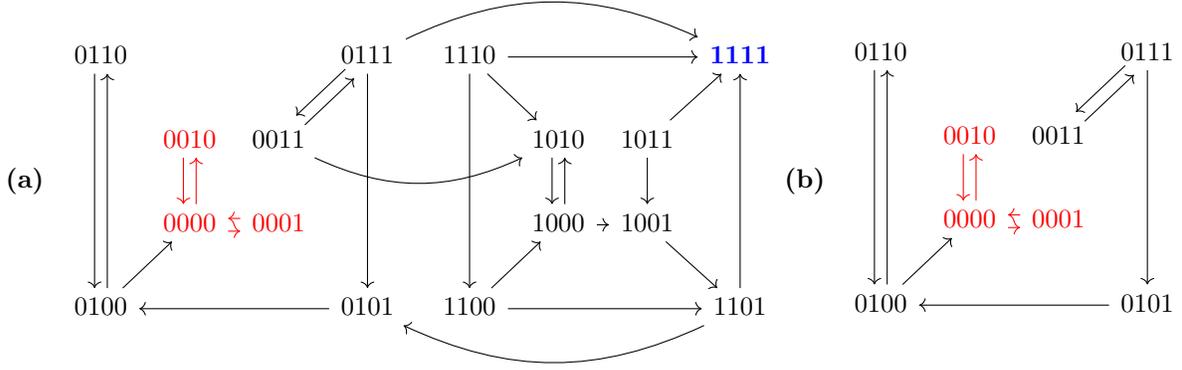


\subsection{Control with model checking}\label{sub:modcheck}

In this section, we present the basis of our new approach for the identification of all the minimal control strategies, based on model checking. To do so, we express the definition of control strategy in terms of CTL formulas. We start by rewriting it in terms of paths.

\begin{lemma}\label{lemma_cs_paths}
Let $f$ be a Boolean function, $\Theta \subseteq \B^n$ a subspace and $P \subseteq \B^n$ a subset. The following are equivalent:

\begin{enumerate}[(i)]
\item $\Theta$ is a control strategy for $P$ in $D(f)$.
\item For every $x \in \Theta$ there exists $y \in P$ such that there exists a path in $D(\fT)$ from $x$ to $y$ and there does not exist any path in $D(\fT)$ from $y$ to any state outside $P$ (that is, all paths starting at $y$ are contained in $P$).
\end{enumerate}
\end{lemma}

\begin{proof}
$(\Rightarrow)$ Let $x \in \Theta$ and let $\A$ be an attractor of $D(\fT)$ that can be reached from $x$. Since $\Theta$ is a control strategy, $\A \subseteq P$. Take $y \in \A$. Since $\A$ is reached from $x$, there exists a path from $x$ to $y \in \A \subseteq P$ and there are no paths from $y$ leaving $P$.

$(\Leftarrow)$ Let $\A$ be an attractor of $\fT$. Let $x \in \A$. Since $x \in \A \subseteq \Theta$, there exists $y \in P$ such that there exists a path in $D(\fT)$ from $x$ to $y$ and there does not exist any path in $D(\fT)$ from $y$ to any state outside $P$. Since $\A$ is an attractor, $y \in \A$ and $\mathrm{Reach}_{D(\fT)}(y) = \A$. Then, $\A \cap \B^n \backslash P = \emptyset$, that is, $\A \subseteq P$, and $\Theta$ is a control strategy for $P$.
\end{proof}

Before expressing \Cref{lemma_cs_paths} in terms of CTL formulas, we introduce a state formula $\psi_S$ that is satisfied if and only if a state $x$ belongs to a subspace $S = \Sigma(I,c)$:
$$
\psi_S(x) = \bigwedge_{i \in I} (x_i = c_i).
$$

This formulation can be extended to subsets as well. Let $P \subseteq \B^n$ be a subset. We define $\phi_P$ to be satisfied if and only if $x \in P$:

$$
\phi_P(x) = \bigvee_{S \in \mathcal{S}} \psi_S(x)
$$

where $\mathcal{S}$ is a subspace cover of $P$.

Now we can express \Cref{lemma_cs_paths} in terms of CTL formulas, using $\phi_P(x)$ as defined above.

\begin{lemma}\label{lemma_ctl}
Let $f$ be a Boolean function, $\Theta \subseteq \B^n$ a subspace and $P \subseteq \B^n$ a subset. The following are equivalent:
\begin{enumerate}[(i)]
\item $\Theta$ is a control strategy for $P$ in $D(f)$.
\item $\Phi_P(x)$, defined as $\Phi_P(x)=$ \textbf{EF}(\textbf{AG}$\phi_P)(x)$, is satisfied in $D(\fT)$ for every $x \in \Theta$.
\end{enumerate}
\end{lemma}

\begin{proof}
$x$ satisfies $\Phi_P$ if and only if there exists a path $x=x^0, x^1, ...$ such that $x^i$ satisfies \textbf{AG}$\phi_P$ for some $i \geq 1$. Let $y = x^i$. $y$ satisfies \textbf{AG}$\phi_P$ if and only if for all paths $y=y^0, y^1, ...$, for all $i \geq 0$, $y^i$ satisfies $\phi_P$, that is, $y^i \in P$. Thus, by \Cref{lemma_cs_paths}, $\Phi_P$ is satisfied for all $x \in \Theta$ if and only if $\Theta$ is a control strategy for $P$ in $D(f)$.
\end{proof}

The CTL formula $\Phi_P$ defined in \Cref{lemma_ctl} provides a way to determine whether a candidate subspace is a control strategy for $P$. The next section presents the implementation of this idea for control strategy identification.


\section{Computation of control strategies} \label{Computation}

Building on the model checking formulas derived in the previous section, we develop a method for control strategy identification. The formula derived in \Cref{lemma_ctl} can be used to define a CTL query that can determine whether a candidate subspace is a control strategy for a target subset. In addition, in order to improve the performance of the method, preliminary checks on the candidate subspace and the restricted network can be conducted to possibly discard it without exhaustive exploration. Moreover, the dimension of the problem, that is, the free variables of the Boolean function, can be reduced by restricting the function to the percolated subspace instead of the candidate subspace, as stated in \Cref{subsec:perc}. The complete implementation of the control strategy identification is detailed in \Cref{alg:cs} and explained in the following.

\begin{algorithm}
\renewcommand{\thealgorithm}{1}
\caption{Control strategies for a target subset P}\label{alg:cs}
\hspace*{\algorithmicindent} \textbf{Input}: $f$ Boolean function, $P$ target subset, $D$ type of update, $m$ size limit \\
\hspace*{\algorithmicindent} \textbf{Output}: control strategies for $P$
\begin{algorithmic}[1]
\Function{ControlStrategies}{$f$, $P$, $D$, $m$}
	\State \textbf{CS} $\gets \emptyset$
	\State \textbf{ST} $\gets \emptyset$ \Comment{ST stores positively checked subspaces}
	\State \textbf{SF} $\gets \emptyset$ \Comment{SF stores negatively checked subspaces}
	\For{$i$ in $\{1, \dots$, min($m$, $n$)$\}$}: \Comment{$n$: number of variables of $f$}
		\State \textbf{S} $\gets$ $\{$S subspace$\colon|$fixed(S)$|$ = $i\}$
		\For {S in \textbf{S}}:
			\If {(S $\not\subseteq$ S' for all S' in \textbf{CS})}:
				\State S$^*$ $\gets$ percolate(\textbf{f}, S)
				\If {S$^* \in$ \textbf{ST}}: add S to \textbf{CS} \EndIf
				\If {S$^* \in$ \textbf{ST} $\cup$ \textbf{SF}}: \textit{break} \EndIf
				\If {S$^* \subseteq P$}:
					\State add S to \textbf{CS}
					\State add S$^*$ to \textbf{ST}
				\Else:				
					\State \textbf{f$^*$} $\gets$ reduce(\textbf{f}, S$^*$)
					\State \textbf{minTS} $\gets$ minimalTrapSpaces(\textbf{f$^*$})
					\If {(T $\cap$ $P$ $\neq \emptyset$ for all T in \textbf{minTS})}:
						\State \textbf{halfTS} $\gets$ $\{$T in \textbf{minTS} if T $\nsubseteq P\}$
						\State valid $\gets$ \textbf{true}
						\For {T in \textbf{halfTS}}:
							\State \textbf{f$^{**}$} $\gets$ reduce(\textbf{f$^*$}, T)
							\State $\Phi_P \gets$ CTLFormula(\textbf{f$^{**}$}, $P$)
							\If {\textbf{not} CTLModelChecking(\textbf{f$^{**}$}, $D$, $\Phi_P$)}:
							\State valid $\gets$ \textbf{false}
							\State add S$^*$ to \textbf{SF}
							\State \textit{break}
							\EndIf
						\EndFor
						\If {valid}:
							\State $\Phi_P \gets$ CTLFormula(\textbf{f$^*$}, $P$)
							\If {CTLModelChecking(\textbf{f$^*$}, $D$, $\Phi_P$)}:
								\State add S to \textbf{CS}
								\State add S$^*$ to \textbf{ST}
							\Else
							\State add S$^*$ to \textbf{SF}
							\EndIf
						\EndIf
					\Else
						\State add S$^*$ to \textbf{SF}
					\EndIf
				\EndIf
			\EndIf
		\EndFor
	\EndFor
	\State \Return \textbf{CS}
\EndFunction
\end{algorithmic}
\end{algorithm}

The main algorithm takes as inputs a Boolean function $f$, a target subset $P$ and the type of update $D$ and returns the minimal control strategies for $P$ in $D(f)$.

For each candidate subspace $S$, its percolated subspace for $\f{S}$, $S^*$, is computed (line 9), as defined in \Cref{def:perc_subs}. By \Cref{perc_cs}, $S^*$ is a control strategy if and only if $S$ is a control strategy, so we perform all the checks on $S^*$. If $S^*$ is contained in the target subset $P$, then $S$ is a control strategy (lines 12-13). If not, the algorithm continues to compute the restriction to $S^*$, $f^*$, and its minimal trap spaces (lines 16-17). If there exists a minimal trap space disjoint from the target, the candidate subspace is discarded, since each minimal trap space contains at least one attractor (line 18). Trap spaces that are partially contained in the target subset are analyzed first (line 19). Since $f(x) = \f{T}(x)$ for all $x \in T$, we can reduce the function to $T$ and run the model checking query for the restriction to $T$, $f^{**}$, (lines 22-24). If the formula is not satisfied in one of these trap spaces, the candidate subspace is discarded. Otherwise, the algorithm concludes by checking the CTL formula $\Phi_P$ for the restricted function $f^*$ and deciding whether $S$ is a control strategy (lines 29-30).

Since the aim is to identify optimal control strategies, the candidate subspaces $S$ are taken randomly fixing an increasing number of variables, so that supersets of sets already defining a successful intervention are not considered (lines 5-8). Furthermore, an upper bound $m$ for the size of the control strategies can be set. Moreover, the decisions made for each percolated subspace are stored in two variables $ST$ (for positively checked subspaces) and $SF$ (for negatively checked subspaces) to avoid repeating the same verification query.

The algorithm presented above is implemented using PyBoolNet \cite{PyBoolNet}, a Python package that allows generation and analysis of Boolean networks, in the module \emph{control$\_$strategies.py}. The source code is available at \url{https://github.com/Lauracf/PyBoolNet/blob/develop/pyboolnet/control_strategies.py}. PyBoolNet uses NuSMV to decide model checking queries for Boolean networks. It also provides an efficient computation of trap spaces for relatively large networks.


\section{Results} \label{Application}

In this section we study the applicability of our method to different biological networks. We start by applying our method to a network modeling the epithelial-to-mesenchymal transition, considering different control targets: attractor, subspace and subset avoidance. In addition, we compare our approach to current control methods in different Boolean networks for attractor and target control. We show that our method is able to identify all the minimal control strategies identified by other approaches, uncovering in some cases minimal control strategies missed by them. All the results presented here were obtained with a regular desktop 8-processor computer, Intel\textsuperscript{\textregistered}Core\textsuperscript{\tiny TM} i7-2600 CPU at 3.40GHz, 16GB memory. The running times vary significantly from scenario to scenario, depending on the type of target considered and the upper bound set on the size of the control strategies. They can range from seconds or minutes, for small control strategies or simple networks, to hours or days of computation for complex target subsets or large control strategies.


\subsection{Case Study: EMT network}

The network considered in this case study was recently introduced by Selvaggio et al. \cite{selvaggio_emt_network} to model how microenvironmental signals influence cancer-related phenotypes along the epithelial-to-mesenchymal transition (EMT). The original network consists of 56 components, ten of them being inputs and two readouts or outputs (see \Cref{fig:selvaggio}). Since the original model is multivalued, we work with its booleanised version obtained with GINsim \cite{GINsim}. This booleanisation maps a multivalued component of maximum value \emph{m} to \emph{m} Boolean components. For instance, a component taking the values 0, 1, 2, 3 is encoded using 3 Boolean variables that would take values 000, 100, 110, 111 respectively (see \Cref{tab:phenotypes}). Although this method introduces states that do not correspond to any value of the multivalued variable (non-admissible states), these cannot be part of any attractor since they always have at least one path to an admissible state and do not have incoming transitions from admissible states. Therefore, the asymptotic behaviour generated strictly replicates the original model. The booleanised network of this case study consists of 60 Boolean variables, whose regulatory functions can be found in the PyBoolNet repository \cite{PyBoolNet}.

\begin{figure}
\centering
\includegraphics[width=0.8\linewidth]{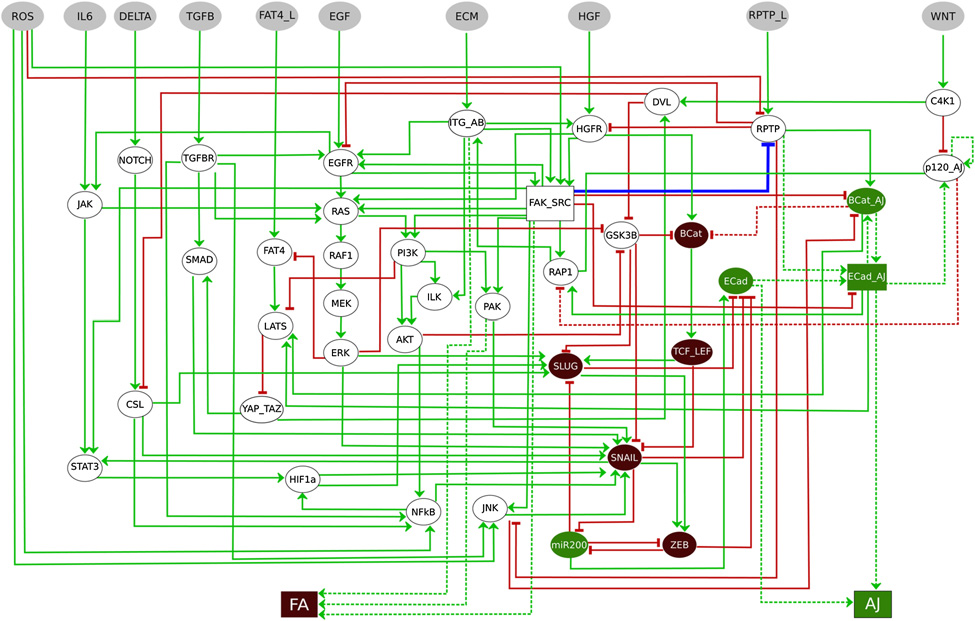}
\caption{EMT multivalued network. Boolean nodes are represented by ellipses and multivalued nodes by rectangles. Input nodes are colored in gray. Image obtained from \cite{selvaggio_emt_network}. Further information about the model can be found in \cite{selvaggio_emt_network}.}
\label{fig:selvaggio}
\end{figure}

The asynchronous dynamics has 1452 attractors, all of them steady states. They are classified according to the values of the readout components \textbf{AJ} and \textbf{FA}, which represent the different degrees of cell adhesions by adherens junctions and focal adhesions respectively \cite{selvaggio_emt_network}. The eight resulting biological phenotypes are divided in four groups: epithelial (E1), mesenchymal (M1, M2, M3), hybrid (H1, H2, H3) and unknown (UN) (see \Cref{tab:phenotypes}).

\begin{table}
\centering
\caption{Relation of the phenotypes of the EMT network, the values of the multivalued readouts (AJ, FA) and the values of the equivalent Boolean components (AJ1, AJ2, FA1, FA2, FA3). The number of steady states belonging to each phenotype is also shown.}
\label{tab:phenotypes}
\begin{tabu}{|l|c|[2pt]c|c|c|[2pt]c|c|c|c|[2pt]c|}
\hline
\multicolumn{2}{|c|[2pt]}{} & \textbf{AJ} & AJ1 & AJ2 & \textbf{FA} & FA1 & FA2 & FA3 & \begin{tabular}{c} Number of \\ steady states \end{tabular} \\
\hline
Epithelial phenotype & E1 & \textbf{2} & 1 & 1 & \textbf{0} & 0 & 0 & 0 & 60 \\
\hline
\multirow{3}{*}{Hybrid phenotypes} & 
H1 & \textbf{2} & 1 & 1 & \textbf{1} & 1 & 0 & 0 & 40 \\
& H2 & \textbf{1} & 1 & 0 & \textbf{2} & 1 & 1 & 0 & 36 \\
& H3 & \textbf{2} & 1 & 1 & \textbf{3} & 1 & 1 & 1 & 48 \\
\hline
\multirow{3}{*}{Mesenchymal phenotypes} & M1 & \textbf{0} & 0 & 0 & \textbf{1} & 1 & 0 & 0 & 208 \\
& M2 & \textbf{0} & 0 & 0 & \textbf{2} & 1 & 1 & 0 & 368 \\
& M3 & \textbf{0} & 0 & 0 & \textbf{3} & 1 & 1 & 1 & 672 \\
\hline
Undefined phenotype & UN & \textbf{0} & 0 & 0 & \textbf{0} & 0 & 0 & 0 & 20 \\
\hline
\end{tabu}
\end{table}

To show the flexibility of our method, we analyse different control targets. We start by targeting single steady states (attractor control). Then, we target the subspaces corresponding to each phenotype (target control). Finally, we study the avoidance of the hybrid phenotype, setting as target the complement of the general hybrid phenotype.

\subsubsection{Attractor control: steady states}

Here we consider the problem of attractor control. Since the control targets are steady states, the minimum number of interventions in each control strategy is at least the number of inputs (each input component needs to be fixed to the corresponding value in the attractor).

For each of the 1452 steady states, the control strategies up to size 13 were identified. 788 steady states have minimal control strategies of size 10, meaning that the dynamics can be controlled to the steady state only by fixing the values of the input components to their values in the attractor. Of the remaining ones, 396 need an extra component to be fixed, 212 require fixing at least two more components and the last 56 steady states require fixing three extra components.

\begin{figure}
\includegraphics[width=\linewidth]{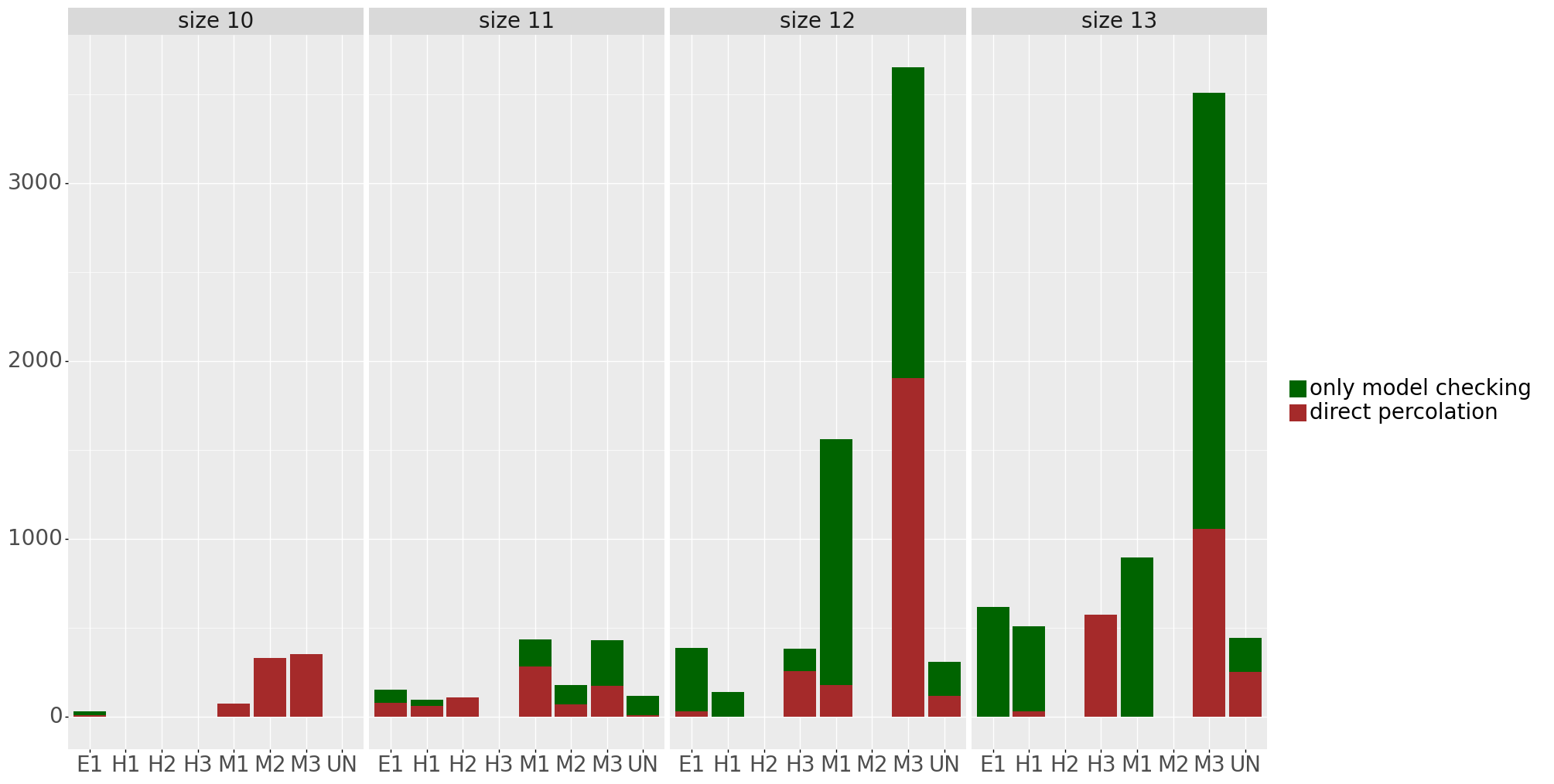}
\caption{Number of control strategies identified for the steady states grouped by phenotype and size. The control strategies obtained by direct percolation are represented in red and the additional control strategies identified by model checking in green. The number of steady states per phenotype can be found in \Cref{tab:phenotypes}.}
\label{barplot_cs}
\end{figure}

\Cref{barplot_cs} shows the number of control strategies identified for each size (10-13) with steady states grouped by phenotype, distinguishing between control strategies identified by direct percolation or only by the model checking approach. In most of the cases, our approach is able to identify many control strategies that are missed by direct percolation. In addition, we observe that the mesenchymal phenotypes are the ones with the highest amount of control strategies, which is to be expected since they are also the ones containing more attractors. The number of steady states per phenotype can be found in \Cref{tab:phenotypes}. Interestingly, no control strategies consisting of only input variables lead to hybrid steady states.

\subsubsection{Target control}

The minimal control strategies up to size 3 are identified for each of the phenotypes, taking as target the subspace defined by the corresponding values of the phenotypic components in each case (see \Cref{tab:phenotypes}). The five phenotypic components are excluded from the candidate interventions, since they represent the readouts of the model that we want to control. \Cref{tab:cs_phenos} shows the number of control strategies identified per phenotype and size. Similarly to the case of attractor control, we observe that the phenotypes with higher number of control strategies are the mesenchymal phenotypes (over a hundred), while the epithelial and the hybrid phenotypes have fewer or no control strategies up to size 3. This is consistent with the bias of the model towards the mesenchymal phenotypes in terms of attractors.

\begin{table}
\centering
\caption{Number of minimal control strategies identified per size for each phenotype. All the control strategies are obtained by direct percolation except three control strategies of size 3 for the phenotype M3 which are only identified by model checking.}
\begin{tabular}{|c|c|c|c|c|c|c|c|c|c|}
\hline
\textbf{Phenotype} & \textbf{E1} & \textbf{H1} & \textbf{H2} & \textbf{H3} & \textbf{M1} & \textbf{M2} & \textbf{M3} & \textbf{UN} \\
\hline
Size 1 & 0 & 0 & 0 & 0 & 0 & 0 & 0 & 0 \\
\hline
Size 2 & 0 & 0 & 0 & 2 & 3 & 3 & 17 & 0 \\
\hline
Size 3 & 0 & 0 & 6 & 2 & 113 & 111 & 83 & 14 \\
\hline
\end{tabular}
\label{tab:cs_phenos}
\end{table}

All of the control strategies obtained are also identified by direct percolation, except for three minimal control strategies for the phenotype M3 that are only identified by model checking. These are: $\{$BCat-AJ = 1, GSK3B = 1, ITG-AB = 1$\}$,  $\{$ECad-AJ1 = 1, GSK3B = 1, ITG-AB = 1$\}$ and $\{$ECad-AJ2 = 1, GSK3B = 1, ITG-AB = 1$\}$.

\subsubsection{Avoidance of hybrid phenotypes}

According to \cite{selvaggio_emt_network}, hybrid phenotypes may provide advantageous abilities to cancer cells such as drug resistance or tumor-initiating potential.
Therefore, interventions avoiding these phenotypes might be good candidates for drug targets in therapeutic treatment against cancer cells presenting these traits.

The authors of \cite{selvaggio_emt_network} define the hybrid phenotype as the one containing steady states with both components AJ and FA activated, that is, AJ $\geq 1$ and FA $\geq 1$. Therefore, the subset defining the avoidance of the hybrid phenotype is

$$
P = \{\mathrm{AJ1} = 0, \mathrm{AJ2} = 0\} \cup \{\mathrm{FA1} = 0, \mathrm{FA2} = 0, \mathrm{FA3} = 0\}.
$$

As in the previous case, the five phenotypic components are excluded from the candidate interventions. Setting the upper bound on the size of the control strategies to 1, 12 control strategies are obtained:
$$
\begin{tabular}{l}
\{\text{ECad} = 0\}, \{\text{ROS} = 1\}, \{\text{SLUG} = 1\}, \{\text{SNAIL} = 1\}, \{\text{TGFB} = 1\}, \{\text{TGFBR} = 1\}, \{\text{ZEB} = 1\}, \\
\{\text{BCat} = 0\}, \{\text{CSL} = 1\}, \{\text{SMAD} = 1\}, \{\text{TCF-LEF} = 0\}, \{\text{miR200} = 0\}.
\end{tabular}
$$

The first seven control strategies can be identified using direct percolation to the maximal subspaces of the target subset. The last five are not identified by using only percolation but are captured by the model checking approach. Two of the obtained interventions correspond to input variables (ROS and TGFB) while the other ten correspond to internal components. Looking at the regulatory functions, we observe that TGFBR is uniquely regulated by TGFB, so setting TGFB to 1 implies that TGFBR is also set to 1 and, therefore, these interventions are equivalent in terms of their effect on phenotypic components. Moreover, since there are no control strategies of size 1 for individual phenotypes, we deduce that these interventions lead to systems where multiple phenotypes coexist, none of them being hybrid.

The components involved in the control strategies identified include the epithelial markers (ECad and miR200) and the mesenchymal ones (BCat, SNAIL, SLUG, TCF-LEF and ZEB) as described in \cite{selvaggio_emt_network}. In addition, the authors of \cite{selvaggio_emt_network} performed a systematic analysis of the effect of single mutants on the attractor landscape, excluding the input variables. All the single mutants corresponding to the non-input interventions found by our approach were identified as having only attractors in non-hybrid phenotypes. Moreover, there was no other single mutation that produced this result. In other words, the results obtained by our approach are in complete correspondence to the ones presented in \cite{selvaggio_emt_network}.


\subsection{Comparison with other methods} \label{subsec:comparison}

In this section, we compare the model checking approach to other control methods currently available. We show that our method is able to capture all the minimal control strategies identified by other methods, uncovering in some cases minimal control strategies that might be missed by them.

In order to be able to compare different approaches, certain common features need to be chosen. Here, we consider control for any possible initial state. We separately compare to methods tackling attractor control and target control. Although approaches for target control can also be used for attractor control when the target attractor is a steady state or minimal trap space, they are usually aimed at targeting larger subspaces, determined for example by a phenotype, which often include several attractors. For this reason, we consider two different scenarios: one for attractor control and one for target control. The case of an arbitrary subset as target could not be considered for comparison, since no other method, to our knowledge, allows this possibility.

The comparison presented here encompasses each of the main approaches for control strategy identification discussed in previous sections (an overview of the main features of the control methods is shown in \Cref{tab:methods}):

\begin{itemize}
\item \textbf{For attractor control:}
	\begin{itemize}
	\item \emph{Stable-motifs approach (SM)}, attractor control method based on the identification of stable motifs as described in \cite{ControlMotifs}.
	\item \emph{Basins approach (BA)}, attractor control method that uses the basin of attraction of the target attractor to identify control strategies as implemented in \cite{cabean}.
	\item \emph{Model checking approach (MC)}, as presented in \Cref{sub:modcheck}.
	\end{itemize}
\item \textbf{For target control:}
	\begin{itemize}
	\item \emph{Percolation-only approach (PO)}, target control method based on percolation into the target subspace as implemented in \cite{Caspo}.
	\item \emph{Trap-spaces approach (TS)}, target control method based on percolation into selected trap spaces introduced in \cite{CS_via_trapspaces}.
	\item \emph{Completeness approach (CN)}, introduced in \Cref{sub:completeness}.
	\item \emph{Model checking approach (MC)}, presented in \Cref{sub:modcheck}.
	\end{itemize}
\end{itemize}

Since the methods for attractor control considered here (BA and SM) only work for asynchronous update, the comparison is only made for this dynamics. In the case of target control, control strategies are identified for both synchronous and asynchronous dynamics.

In view of the different nature of each method, we do not compare their running times. Some approaches require the computation of the system attractors to identify the control strategies (BA, SM). Others do not allow to choose a certain attractor as target and identify control strategies for all the attractors simultaneously (SM). Therefore, a fair comparison with respect to the running times is hard to achieve. For this reason, we focus on the amount and size of the control strategies identified, provided that the program terminates within a few hours.

In order to capture different control scenarios, several biological networks of different sizes with different type and number of attractors are considered. A short description of each network is provided below. See \Cref{tab:networks} for an overview of the networks and their features. The Boolean rules for each biological network can be found in the PyBoolNet repository \cite{PyBoolNet}.
 
\begin{enumerate}[(a)]

\item T-LGL network, introduced by Zhang et al. (2008) \cite{TLGL_network} to model the T cell large granular lymphocyte (T-LGL) survival signaling network. In order for SM to terminate the processing of the network within a few hours, it has been adapted as in \cite{ControlMotifs}, removing the outgoing interactions of \emph{Apoptosis} and setting \emph{Stimuli} and \emph{IL15} to 1 and the remaining inputs to 0. The simplified network consists of 60 Boolean variables and its asynchronous dynamics has 3 cyclic attractors, versus the 156 (steady states and cyclic) of the original. For sake of simplicity, the same modified network is used for attractor control and target control.
 
\item MAPK, introduced by Grieco et al. (2013) \cite{MAPK_network} to model the effect of the Mitogen-Activated Protein Kinase (MAPK) pathway on cell fate decisions taken in pathological cells. The network consists of 53 Boolean variables and it has 18 attractors in the asynchronous dynamics, 12 steady states and 8 cyclic attractors.

\item Cell-Fate network, introduced by Calzone et al. (2010) \cite{CellFate_network} to model the cell fate decision process. The network uses 28 Boolean variables and its asynchronous dynamics has 27 attractors, all of them steady states. These are classified in four different phenotypes (Apoptosis, Survival, Non-Apoptotic Cell Death and Naive) according to the values of the output components of the network.

\end{enumerate}

\begin{table}
\caption{Overview of the versatility of the different control methods in terms of the types of targets and update schemes.}
\adjustbox{scale=0.75, center}{
\begin{tabu}{|lc|c|[2pt]c|c|c|c|c|[2pt]c|c|c|}
\hline
\multicolumn{3}{|c|[2pt]}{} & \multicolumn{5}{c|[2pt]}{\textbf{Control target}} &  \multicolumn{3}{c|}{\textbf{Update}} \\
\hline
\multicolumn{2}{|c|}{\textbf{Method}} & \textbf{Tool} & \begin{tabular}{c} steady \\ state \end{tabular} & \begin{tabular}{c} trap space \\ attractor \end{tabular} & \begin{tabular}{c} complex \\ attractor \end{tabular} & subspace & \begin{tabular}{c} arbitrary \\ subset \end{tabular} & async & sync & \begin{tabular}{c} general \\ async \end{tabular} \\
\hline
Basins & BA & CABEAN \cite{cabean} & \checkmark & \checkmark & \checkmark & - & - & \checkmark & - & - \\
Stable Motifs & SM & StableMotifs \cite{ControlMotifs} & \checkmark & \checkmark & - & - & - & \checkmark & - & - \\
Percolation-only & PO & Caspo \cite{Caspo} & \checkmark & \checkmark & - & \checkmark & - & \checkmark & \checkmark & \checkmark \\
Trap spaces & TS &  PyBoolNet \cite{CS_via_trapspaces} & \checkmark & \checkmark & - & \checkmark & - & \checkmark & \checkmark & \checkmark \\
Completeness & CN &  PyBoolNet & \checkmark & \checkmark & - & \checkmark & - & \checkmark & \checkmark & \checkmark \\
Model Checking & MC &  PyBoolNet & \checkmark & \checkmark & \checkmark & \checkmark & \checkmark & \checkmark & \checkmark & \checkmark \\
\hline
\end{tabu}
}
\label{tab:methods}
\end{table}


\begin{table}
\centering
\caption{Main features of the biological networks used in the comparison. Input components are fixed in the T-LGL network and free in the Cell-Fate and MAPK networks, unless specified otherwise.}
\begin{tabular}{|ll|c|c|c|c|c|}
\hline
\multicolumn{2}{|c|}{\multirow{2}*{\textbf{Network}}} & \multirow{2}*{\textbf{Size}} & \multirow{2}*{\textbf{Inputs}} & \multirow{2}*{\textbf{Outputs}} & \multicolumn{2}{c|}{\textbf{Attractors}} \\
\cline{6-7}
 &  &  &  & & steady & cyclic \\
\hline
Cell-Fate & \cite{CellFate_network} & 28 & 3 & 3 & 27 & 0 \\
MAPK & \cite{MAPK_network} & 53 & 4 & 3 & 12 & 6 \\
T-LGL & \cite{TLGL_network} & 60 & 6 & 3 & 0 & 3 \\
\hline
\end{tabular}
\label{tab:networks}
\end{table}

\subsubsection{Attractor control}

We compare our model checking approach (MC) to two methods for attractor control: stable motifs (SM) \cite{ControlMotifs} and basins of attraction (BA) \cite{cabean}. SM works for steady states and the complex attractors captured by the stable motifs (in some cases, complex attractors are not identified and the method cannot be applied). BA works for any kind of attractors and computes their basins of attraction to identify minimal control strategies.

\begin{table}
\centering
\caption{Number and size of the control strategies up to size 4 identified by each method (SM, BA and MC) for the corresponding attractor of each biological network, with the input components fixed as mentioned in the main text. When a method obtains non-minimal control strategies, the number of minimal control strategies identified is indicated in parenthesis. Note that BA does not look for larger control strategies once a minimal one (with respect to size) is obtained.}
\begin{tabular}{|l|c|c|c|c|c|}
\hline
\textbf{Network} & \begin{tabular}{l} \textbf{Method} \end{tabular} & \textbf{Size 1} & \textbf{Size 2} & \textbf{Size 3} & \textbf{Size 4} \\
\hline
Cell-Fate & \begin{tabular}{c} SM \\ BA \\ MC \end{tabular} &
 \begin{tabular}{c} 0 \\ 0 \\ 0 \end{tabular} &
 \begin{tabular}{c} 2 \\ 6 \\ 8 \end{tabular} & 
 \begin{tabular}{c} 28(12) \\ 0 \\ 28 \end{tabular} & 
 \begin{tabular}{c} 2(0) \\ 0 \\ 0 \end{tabular} \\
\hline
T-LGL & \begin{tabular}{c} SM \\ BA \\ MC \end{tabular} & 
\begin{tabular}{c} 3 \\ 4 \\ 4 \end{tabular} & 
\begin{tabular}{c} 0 \\ 0 \\ 0 \end{tabular} & 
\begin{tabular}{c} 5(0) \\ 0 \\ 0 \end{tabular} & 
\begin{tabular}{l} 1(0) \\ 0 \\ 0 \end{tabular} \\
\hline
MAPK & \begin{tabular}{l} SM \\ BA \\ MC \end{tabular} & 
\begin{tabular}{c} 0 \\ 0 \\ 0 \end{tabular} & 
\begin{tabular}{c} 0 \\ 0 \\ 0 \end{tabular} & 
\begin{tabular}{c} 0 \\ 2 \\ 2 \end{tabular} & 
\begin{tabular}{c} 16(0) \\ 0 \\ 0 \end{tabular} \\
\hline
\end{tabular}
\label{tab:cs_attr}
\end{table}


\begin{table}
\centering
\caption{Minimal control strategies up to size 4 identified by each method (SM, BA and MC) for the selected attractor of each biological network. $\checkmark$ and $-$ denote whether the control strategy is obtained by the method or not, respectively. For simplicity, only the control strategies of size 2 are included for the Cell-Fate network.}
\begin{tabular}{|l|l|c|c|c|}
\hline
\textbf{Network} & \begin{tabular}{l} \textbf{Minimal Control Strategies} \end{tabular} & \textbf{SM} & \textbf{BA} & \textbf{MC} \\
\hline
Cell-Fate & \begin{tabular}{l} $\{$BAX = 1, MPT = 0$\}$ \\ 
$\{$BAX = 1, ROS = 0$\}$ \\ $\{$CASP3 = 1, MPT = 0$\}$ \\
$\{$CASP3 = 1, ROS = 0$\}$ \\ $\{$CASP8 = 1, MPT = 0$\}$ \\
$\{$CASP8 = 1, ROS = 0$\}$ \\ $\{$MOMP = 1, MPT = 0$\}$ \\
$\{$MOMP = 1, ROS = 0$\}$ \\ \end{tabular} &
\begin{tabular}{c} - \\ - \\ \checkmark \\ \checkmark \\ - \\ - \\ - \\ - \\ \end{tabular} &
\begin{tabular}{c} \checkmark \\ \checkmark \\ \checkmark \\ - \\ \checkmark \\ \checkmark \\ \checkmark \\ - \\ \end{tabular} &
\begin{tabular}{c} \checkmark \\ \checkmark \\ \checkmark \\ \checkmark \\ \checkmark \\ \checkmark \\ \checkmark \\ \checkmark \\ \end{tabular} \\
\hline
T-LGL & \begin{tabular}{l} $\{$Ceramide = 1$\}$ \\ $\{$PDGFR = 0$\}$ \\ $\{$S1P = 0$\}$ \\ $\{$SPHK1 = 0$\}$ \\
\end{tabular} & \begin{tabular}{c} - \\ \checkmark \\ \checkmark \\ \checkmark \\ \end{tabular} &
 \begin{tabular}{c} \checkmark \\ \checkmark \\ \checkmark \\ \checkmark \\ \end{tabular} & 
 \begin{tabular}{c} \checkmark \\ \checkmark \\ \checkmark \\ \checkmark \\ \end{tabular} \\
\hline
MAPK & \begin{tabular}{l} $\{$DNA-damage = 1, TGFBR-stimulus = 0, GAB1 = 0$\}$ \\ $\{$DNA-damage = 1, TGFBR-stimulus = 0, PI3K = 0$\}$ \\ \end{tabular} &
\begin{tabular}{c} - \\ - \\ \end{tabular} & \begin{tabular}{c} \checkmark \\ \checkmark \\ \end{tabular} & \begin{tabular}{c} \checkmark \\ \checkmark \\ \end{tabular} \\
\hline
\end{tabular}
\label{tab:compare_attr}
\end{table}

The control problems selected for each network are described below.

\begin{enumerate}[(a)]

\item \textbf{T-LGL network.} The three attractors of this network can be classified in two types (Survival and Apoptosis) according to the values of the output components. For this comparison, the apoptotic attractor is chosen, that is, the one with Apoptosis = 1. Similar results would be obtained for another choice of attractor.

\item \textbf{MAPK network.} The eighteen attractors of this network can also be classified in two types (Survival and Apoptosis) according to the values of the output components. In order to apply BA to this network, we set some input components to fixed values. We consider all the input combinations that allow the two phenotypes to coexist. Three input-value combinations satisfy this condition. For sake of space, we only show the results for one of the combinations, the one fixing EGFR-stimulus = 0 and FGFR3-stimulus = 0 and targeting an apoptotic attractor. Similar results for BA and MC would be obtained for the other input-value combinations and different choices of attractor. Results for SM might vary depending on the target attractor that is chosen.

\item \textbf{Cell-Fate network.} The 28 attractors of this network can be classified in four different phenotypes (Apoptosis, Survival, Non-Apoptotic Cell Death and Naive) according to the values of the output components. The control strategies up to size 4 obtained by each method for each of the attractors of the network are the same, except for five attractors, where the SM approach missed some of the minimal control strategies obtained by BA and MC. In order to gain more insight, we also compare the results when fixing the input components. We analyse the five input-value combinations that allow the three relevant phenotypes (Apoptosis, Survival and NonACD) to coexist. For sake of space, we only show the results for one of the combinations, the one fixing FADD = 0, FASL = 1 and TNF = 1 and targeting the apoptotic attractor. Similar results would be obtained for the other input-value combinations and different choices of attractor.

\end{enumerate}

\Cref{tab:cs_attr} contains the size and number of the control strategies up to size 4 computed by each approach for each network. MC is able to identify all the minimal control strategies for every network. SM obtains some non-minimal control strategies of larger size, which are supersets of minimal ones. BA usually identifies all the minimal control strategies of minimum size, as done by MC, but it misses two minimal control strategies of size 2 for the Cell-Fate network. The minimal control strategies for each network and the methods that are able to identify them are shown in \Cref{tab:compare_attr}. While BA does not identify any control strategy larger than the minimum size, MC computes all the strategies minimal with respect to inclusion.

\subsubsection{Target control}

We compare our model checking approach (MC) in asynchronous and synchronous dynamics to several methods for target control: percolation to target (PO) \cite{Caspo}, percolation via trap spaces (TS) \cite{CS_via_trapspaces} and the completeness approach (CN), developed in \Cref{sub:completeness}. An overview of the features of each control method is shown in \Cref{tab:methods}.

\begin{table}
\centering
\caption{Number and size of the control strategies identified by each method (PO, TS, CN and MC) for the corresponding attractor of each biological network in the asynchronous and synchronous dynamics. When a method obtains non-minimal control strategies, the number of minimal control strategies identified is indicated in parenthesis.}
\begin{tabu}{|l|c|[2pt]c|c|c|[2pt]c|c|c|}
\hline
\multicolumn{2}{|c|[2pt]}{} & \multicolumn{3}{c|[2pt]}{Asynchronous} & \multicolumn{3}{c|}{Synchronous} \\
\hline
\textbf{Network} & \begin{tabular}{l} \textbf{Method} \end{tabular} & \textbf{Size 1} & \textbf{Size 2} & \textbf{Size 3} & \textbf{Size 1} & \textbf{Size 2} & \textbf{Size 3} \\
\hline
Cell-Fate & \begin{tabular}{c} PO \\ TS \\ CN \\ MC \end{tabular} & 
\begin{tabular}{c} 0 \\ 0 \\ 0 \\ 0 \end{tabular} & 
\begin{tabular}{c} 17 \\ 17 \\ 21 \\ 21 \end{tabular} & 
\begin{tabular}{c} 173 \\ 173 \\ 191 \\ 191 \end{tabular}
 &  
\begin{tabular}{c} 0 \\ 0 \\ 0 \\ 0 \end{tabular} & 
\begin{tabular}{c} 17 \\ 17 \\ 17 \\ 17 \end{tabular} & 
\begin{tabular}{c} 173 \\ 173 \\ 189 \\ 189 \end{tabular} \\
\hline
T-LGL & \begin{tabular}{c} PO \\ TS \\ CN \\ MC \end{tabular} & 
\begin{tabular}{c} 0 \\ 3 \\ 10 \\ 10 \end{tabular} & 
\begin{tabular}{c} 224 (77) \\ 195 (77) \\ 116 \\ 116 \end{tabular} & 
\begin{tabular}{c} 327 (77) \\ 282 (77) \\ 204 \\ 204 \end{tabular}
 &  
\begin{tabular}{c} 0 \\ 0 \\ 0 \\ 3 \end{tabular} & 
\begin{tabular}{c} 224 (164) \\ 224 (164) \\ 232 (172) \\ 251 \end{tabular} & 
\begin{tabular}{c} 327 (97) \\ 327 (97) \\ 762 (109) \\ 261 \end{tabular} \\
\hline
MAPK & \begin{tabular}{c} PO \\ TS \\ CN \\ MC \end{tabular} & 
\begin{tabular}{c} 2 \\ 3 \\ 8 \\ 8 \end{tabular} & 
\begin{tabular}{c} 124 (59) \\ 106 (59) \\ 105 \\ 105 \end{tabular} & 
\begin{tabular}{c} 175 (45) \\ 162 (45) \\ 66 \\ 66 \end{tabular}
 & 
\begin{tabular}{c} 2 \\ 2 \\ 2 \\ 4 \end{tabular} & 
\begin{tabular}{c} 124 (88) \\ 124 (88) \\ 164 (112) \\ 118 \end{tabular} & 
\begin{tabular}{c} 175 (49) \\ 175 (49) \\ 195 (155) \\ 216 \end{tabular} \\
\hline
\end{tabu}
\label{tab:cs_target}
\end{table}

The target subspaces chosen for each network are the ones corresponding to the apoptotic phenotype, a common target in drug identification studies for cancer therapeutic treatments \cite{apoptosis_cancer}. They are defined in terms of the output components of each network:
\begin{enumerate}[(a)]
\item \textbf{T-LGL:} $\{$Apoptosis = 1, Proliferation = 0$\}$.
\item \textbf{MAPK:} $\{$Apoptosis = 1, Proliferation = 0, Growth-Arrest = 1$\}$.
\item \textbf{Cell-Fate:} $\{$Apoptosis = 1, Survival = 0, NonACD = 0$\}$.
\end{enumerate}

\Cref{tab:cs_target} contains the size and number of the control strategies computed by each approach for the asynchronous and synchronous dynamics. It is important to note that all the minimal control strategies obtained by PO are included in the ones identified by TS, since TS is built on top of PO. Moreover, direct percolation is a pre-check for CN and MC methods and, therefore, all minimal control strategies found by PO are obtained by CN and MS as well. PO is able to identify a high number of minimal control strategies, as is to be expected since regulatory functions modeling biological systems usually induce a lot of percolation. Nonetheless, in some networks this number is still far from the number of minimal control strategies identified by MC.

Control strategies are update-dependent by definition. However, all the control strategies identified by PO are valid in all the dynamics considered here. The number of these control strategies that are minimal might vary from one update to another (see results for T-LGL and MAPK networks). On the other hand, methods TS, CN and MC are sensitive to the update. In the case of TS, since none of the networks is complete in the synchronous dynamics, attractors cannot be approximated by minimal trap spaces and, therefore, no additional control strategies compared to PO are obtained. Interestingly, the methods CN and MC obtain the same number of control strategies for the asynchronous dynamics. This is not the case for the synchronous dynamics, where additional control strategies are obtained by MC for the T-LGL and MAPK networks. This is caused by the fact that the CN method cannot classify a subspace as a control strategy if the restricted network is not complete. The CN method is likely to obtain better results when the original network is complete, which is usually the case for the asynchronous dynamics of biological networks \cite{AttractorApprox_Klarner}.

These case studies illustrate how an exhaustive approach like the one provided by the model checking method introduced in this paper has the capacity to identify, in networks of practical relevance, simple sets of interventions that might otherwise not be identified, and that can furnish additional insights into the model, widening the possibility for potential applications.


\section{Discussion}

This work presents a novel method to identify all the minimal control strategies for an arbitrary target subset. It is able to deal not only with the problems of attractor control and target control, already tackled by existing methods, but also with subset control, providing maximal flexibility on the definition of the control goal and allowing for instance the possibility of dealing with attractor avoidance problems.

The comparison performed in \Cref{subsec:comparison} shows that our approach is able to identify all the minimal control strategies obtained by the methods analysed and, in many cases, to uncover new control strategies. It also provides flexibility to study different control problems, which can lead to additional insights into the network. For these reasons, the method presented here can be a good option when a deep analysis of the control strategies of the model is required.

Even though running times are not compared in this work, the model checking approach is likely to entail more computational time than other approaches due to its exhaustiveness, since in some cases the full exploration of the state space might be required. Although the use of symbolic states allows to deal with relatively large networks, the computational time required to explore their state spaces might still be too high. For this reason, several steps have been developed to reduce the candidate space and the dimensionality of the problem, improving the overall performance. The case studies of \Cref{Application} show the applicability of the method to real models of interest. In cases where this reduction might still be insufficient, our approach could be used to complement faster methods for the particular scenarios in which a more exhaustive analysis is needed.

This new approach is able to adapt to different types of targets and updates. This flexibility, provided by the versatility of model checking, results in the potential to be extended to other control problems. Possible future works could include its extension to tackle control from a set of initial states or the addition of other types of interventions.

\section*{Acknowledgments}

We would like to thank Claudine Chaouiya, Florence Janody and Hannes Klarner for fruitful discussions. 

\section*{Funding}

LCF was partially funded by the Volkswagen Stiftung (Volkswagen Foundation) under the funding initiative Life? - A fresh scientific approach to the basic principles of life (Project ID: 93063).

\bibliographystyle{abbrv}
\bibliography{references}

\end{document}